\documentclass[11pt,a4paper]{article}

\usepackage[T1]{fontenc}

\usepackage{citesort,graphicx, algorithm,algorithmic,clrscode3e,tikz,bm,caption,amssymb,amsmath,empheq,enumitem,graphicx,color,footmisc,appendix,fullpage}
\usepackage{amsthm}
\usepackage{caption}
\usepackage{subcaption}
\usepackage{tikz}
\usetikzlibrary{calc}
\usetikzlibrary{arrows.meta}

\def\B {\textit{$O$}}  
\def\I {\textit{$I$}}

\newtheorem{theorem}{Theorem}[section]
\newtheorem{lemma}[theorem]{Lemma}
\newtheorem{Lemma}[theorem]{Lemma}

\usepackage{enumitem}
\newlist{Property}{enumerate}{2}
\setlist[Property]{label=Property \arabic*., font=\textbf, itemindent=*}

\newcommand{\rgb}[1]{\textcolor[rgb]{0.00,0.00,0.00}{#1}}

\date{}

\begin{document}
\title{Complexity of Paired Domination Problems on Circle and $k$-Polygon Graphs}
	
\vspace{30pt}

\author{
    Ta-Yu Mu\thanks{
        Department of Computer Science and Information Engineering,
        National Taiwan University,
        Taipei 10617, \newline \mbox{} \hspace{11pt} Taiwan.
        Email: f08922132@ntu.edu.tw}
    \and
    Ching-Chi Lin\thanks{
        Department of Computer Science and Engineering,
        National Taiwan Ocean University,
        Keelung 20224, \newline \mbox{} \hspace{11pt} Taiwan. Corresponding author.
        Email: lincc@mail.ntou.edu.tw}
}

\maketitle              
\begin{abstract}
A set $D \subseteq V$ is a dominating set of a graph $G$ if every vertex in $V - D$ is adjacent to at least one vertex in $D$. A dominating set $D$ is a paired-dominating set if the subgraph of $G$ induced by $D$ contains a perfect matching. In this paper, we prove that determining the minimum paired-dominating set in circle graphs is NP-complete. We further present an $O(n(\frac{n}{k^2-k})^{2k^2-2k})$-time algorithm for finding the minimum paired-dominating set in $k$-polygon graphs, a subclass of circle graphs. Additionally, we refine the existing algorithm of Elmallah and Stewart for computing the minimum dominating set in $k$-polygon graphs, reducing its time complexity from $O(n^{4k^2+3})$ to $O(n^{3k-5})$, and further extend it to find the minimum total dominating set.

\bigbreak

\noindent \textbf{Keywords:} {circle graph; $k$-polygon graph; domination; paired domination; total domination; NP-complete; Hamiltonian path.}
\end{abstract}

\baselineskip 19pt
\newpage
\section{Introduction}
Finding a dominating set is a fundamental problem in graph theory with numerous applications, including network design, facility location, and social network analysis~\cite{Haynes98-2,Haynes98}. Given a graph $G = (V, E)$, a dominating set $D \subseteq V$ is defined such that every vertex $v \in V - D$ is adjacent to at least one vertex in $D$. The minimum dominating set problem seeks to find a dominating set of minimum size. Over the past few decades, various extensions of domination have been explored, such as total domination~\cite{Cockayne80}, connected domination [28], independent domination~\cite{Goddard13}, and paired-domination~\cite{Haynes98-3}. In total domination, the induced subgraph $G[D]$ contains no isolated vertices, while in connected domination, $G[D]$ must be connected. Independent domination requires that the dominating set also forms an independent set.

A set $M \subseteq E$ is called a matching if no two edges in $M$ share a common vertex. A matching $M$ is perfect if every vertex in $V$ is covered by an edge in $M$. The concept of paired-domination requires that the induced subgraph $G[D]$ contains a perfect matching. Paired-domination was first introduced in~\cite{Haynes98-3} in the context of an area monitoring problem, where a minimum even number of guards were stationed at various sites to monitor all locations. Each guard was paired with another at a neighboring site, ensuring mutual backup. This concept has since found applications in various domains, including security and surveillance, resource allocation, facility location, network communication, and coding theory~\cite{Gupta13,Dreyer00,Borzooei15}.

The minimum paired-domination problem seeks to find a minimum dominating set $D$ such that $G[D]$ contains a perfect matching. Its NP-completeness for general graphs was established in~\cite{Haynes98-3}. Figure~\ref{figure:domination} illustrates an example of a minimum dominating set, a minimum paired-dominating set, and a minimum total dominating set. As below, we review prior research on paired domination.

In~\cite{Amjadi22}, the minimum subdivision number of edges was determined to increase the size of the minimum paired-dominating set. The sizes of minimum dominating sets, minimum total dominating sets, and minimum paired-dominating sets were analyzed in~\cite{Dettlaff22,Eakawinrujee23}. The concept of paired domination stability was introduced in~\cite{Gorzkowska23}, referring to the minimum non-isolated vertex set whose removal decreases the size of a paired-dominating set. Inspired by paired-domination, a strategic game involving two roles—dominator and staller—was proposed in~\cite{Gray23}, where the dominator aims to minimize the paired-dominating set while the staller seeks to maximize it. In~\cite{Henning20}, the maximum sizes of minimal paired-dominating sets were studied for specific graph classes, including bipartite graphs, threshold graphs, chain graphs, and proper interval graphs.

\begin{figure}
	\begin{center}
		\begin{tikzpicture}[scale=0.7]
			            \tikzset{
                vertex/.style ={circle, draw=black, fill=white, thick, minimum size=7mm, inner sep=2pt}
            }
            \draw node[vertex] (a) at(0, 5) {$a$};
            \draw node[vertex] (d) at(-1.5, 2) {$d$};
            \draw node[vertex] (e) at(1.5, 2) {$e$};
            \draw node[vertex] (b) at(-4, 3) {$b$};
            \draw node[vertex] (c) at(4, 3) {$c$};
            \draw node[vertex] (f) at(0, 0) {$f$};
            \draw node[vertex] (h) at(-3, -2) {$h$};
            \draw node[vertex] (g) at(3, -2) {$g$};
            \draw [thick] (a)--(b);
            \draw [thick] (a)--(c);
            \draw [thick] (a)--(d);
            \draw [thick] (a)--(e);
            \draw [thick] (f)--(d);
            \draw [thick] (f)--(e);
            \draw [thick] (f)--(g);
            \draw [thick] (f)--(h);
            \draw [thick] (h)--(b);
            \draw [thick] (h)--(d);
            \draw [thick] (g)--(e);
            \draw [thick] (g)--(h);
		\end{tikzpicture}
	\end{center}
	\vspace{-10pt}
	\caption{A graph with a minimum dominating set $\{a,g\}$, a minimum paired-dominating set $\{a,c,f,h\}$, and a minimum total dominating set $\{a,e,g\}$, .
	}
	\label{figure:domination}
\end{figure}

In~\cite{Goddard09}, a conjecture: given a connected graph of n vertices, each having degree at least three, its minimum paired-dominating set had size not greater than $\frac{4n}{7}$, if it was not the Petersen graph, was suggested. It was later verified for $r$-regular graphs in~\cite{Lu16} with $r \ge 4$ and in~\cite{Kosari22} with $r = 3$. 
In~\cite{Tuncel23}, the sizes of minimum paired-dominating sets were compared for two special families of transformation graphs. An $O(n^{4k^2+3})$-time algorithm for finding a minimum dominating set on a $k$-polygon graph was proposed in~\cite{ElMallah93}. Algorithms for computing minimum paired-dominating sets can be found in~\cite{Mu22,Lappas13}. Specifically, an $O(n + m)$-time algorithm was developed for distance-hereditary graphs in~\cite{Mu22}, while~\cite{Lappas13} presented a linear-time algorithm for permutation graphs, where $n$ and $m$ denote the number of vertices and edges, respectively. In~\cite{Lin22}, considering a weighted setting where each vertex is assigned a positive weight, an $O(n+m)$-time algorithm was proposed for finding a paired-dominating set of a block graph that minimizes the total vertex weight.

A circle graph $G = (V, E)$ represents an intersection model of chords positioned on a circle, where each vertex $v \in V$ corresponds uniquely to a chord, and an edge $(u, v) \in E$ exists if and only if the chords corresponding to $u$ and $v$ intersect. If the chords are positioned on a $k$-polygon, the resulting graph is a $k$-polygon graph. In fact, the graph shown in Figure~\ref{figure:circle_and_polygon} is both a circle graph and a 4-polygon graph. Their corresponding intersection models are shown in Figure~\ref{figure:circle_and_polygon}. It is important to note that $k$-polygon graphs form a subclass of circle graphs, as demonstrated in~\cite{Elmallah98}. Additionally, an $O(4^kn^2)$-time algorithm was proposed to determine whether a given circle graph is a $k$-polygon graph.

\begin{figure}
	\begin{center}
		\begin{tikzpicture}[scale=0.6]
			\input{circle_model}
		\end{tikzpicture}
		\hspace{1.5cm}
        \begin{tikzpicture}[scale=0.6]
			\input{polygon_model}
		\end{tikzpicture}
	\end{center}
\caption{Two intersection models for the graph shown in Figure~\ref{figure:domination}.}
\label{figure:circle_and_polygon}
\end{figure}

In this paper, we first prove that finding a minimum paired-dominating set on a circle graph is NP-complete. Next, we propose two polynomial-time algorithms: one for finding minimum dominating sets in $O(\frac{n^{2k^2-4k+1}}{(k^2-k)^{k^2-k} (k^2-3k)^{k^2-3k}})$ time, and the other for finding minimum paired-dominating sets in $O(n(\frac{n}{k^2-k})^{2k^2-2k})$ time, both for $k$-polygon graphs, a subclass of circle graphs. Furthermore, we improve the previous algorithm for the domination problem, reducing its time complexity to $O(n^{3k-5})$. Our results enhance the work of~\cite{ElMallah93}. In the next two sections, we present the NP-hardness proof and the two proposed algorithms. Finally, some concluding remarks are provided in Section~\ref{section:conclusion}.

 

\smallskip
\section{NP-completeness for Circle Graphs}
\label{section:circle}
We begin by introducing our notations. Let the chord set $J(G)$, or simply $J$ if $G$ is unambiguous, denote the intersection model of $G=(V,E)$. For each \rgb{pair of} chords $c_i,c_j\in J$, corresponding to the vertices $v_i,v_j\in V$, respectively, $c_i$ intersects $c_j$ if and only if $(v_i,v_j)\in E$. We further define the close neighbor on $J$ as $N_J[c]=\{d\mid d\text{ intersects }c\}\cup\{c\}$ and $N_J[S]=\bigcup_{c\in S} N_J[c]$. For the sets $D\subseteq J$ and $M\subseteq J\times J$, \rgb{when} unambiguous, we also call $M$ a perfect matching of $J[D]$ if $G[\{v_i\mid c_i\in V\}]$ simultaneously contains a perfect matching $\{(v_i,v_j)\mid (c_i,c_j)\in M\}$.

In our NP-completeness result, we choose to reduce from the Hamiltonian path problem on directed graphs, a well-known NP-complete problem for finding a path visiting each vertex of the graph exactly once, known as a Hamiltonian path~\cite{Garey76}. The decision formulation of this problem is defined as below:



\noindent\textbf{Hamiltonian path problem}\\
\texttt{INPUT}: A directed graph $G_H$ with the ~designated starting vertex and ending vertex.\\
\texttt{OUTPUT}: Determine whether there exists a Hamiltonian path in $G_H$ from the starting vertex to the ending vertex.\\

In our reduction, we construct a circle graph $G_c$ in polynomial time corresponding to any instance $G_H$ of Hamiltonian path problem. Let $G_H=(V_H,E_H)$ where $|V_H|=n$ and $|E_H|=m$. The following theorem is involved in establishing our NP-completeness result, where the correctness \rgb{is} ensured by Lemmas~\ref{lemma:Hamiltonian_->_MPDS} and~\ref{lemma:Hamiltonian_<-_MPDS}.

\begin{theorem}\label{theorem:Hamiltonian_to_MPDS}
	The graph $G_H$ contains a Hamiltonian path if and only if $G_c$ has a paired-dominating set with cardinality $2n^2+2n-2$.
\end{theorem}

\smallskip

Now we begin introducing the construction of the $G_c$ based on the given instance $G_H$ of the Hamiltonian path problem. Without loss of generality, we may assume the starting vertex as $v_1$ and the ending vertex as $v_n$. Instead of presenting $G_c$ directly, we detail the structure of its chord intersection model $J$ as below. Fig~\ref{figure:circle_reduction_part} illustrates the structure of $J$ for the instance $G_H$ where $V_H=\{v_1,v_2,v_3\}$ and $E_H=\{(v_1,v_3),(v_2,v_3)\}$ as an example.

The chord set $J$ consists of eight types of chords: type-$\ell$, type-$c$, type-$r$, type-$e$, type-$f$, type-$f'$, type-$a$, type-$a'$, type-$b$, and type-$b'$ chords. The type-$\ell$, type-$c$, and type-$r$ chords are defined relative to $V_H$. Specifically, for each vertex $v_i\in V_H$ (where $1\le i\le n$), the corresponding chords $\ell_i^j$, $c_i^j$, and $r_i^j$ exist for $1 \le j \le n$, with the exception that $\ell_1^1$ and $r_n^n$ do not exist. 
The type-$\ell$ and type-$r$ chords are pairwise non-intersection and arranged in counterclockwise order from $i=1$ to $n$ and from $j=1$ to $n$. The type-$c$ chord intersects both the type-$\ell$ and type-$r$ chords corresponding to the same $i,j$ pair. Fig.~\ref{figure:circle_reduction_part}$(a)$ provides an example of the types of chords mentioned.

For each edge $(v_x,v_y)\in E(G_H)$, there exists corresponding type-$e$ chords $e_{x,y}^j\in J$ for $1\le j<n$ intersecting both $r_x^j$ and $\ell_y^{j+1}$. We omit the description of intersections between pairs of type-$e$ chords, as their intersections are irrelevant for our reduction. 
Intersecting with those type-$e$ chords $e_{x,y}^j$ with the same superscript, the type-$f$ chord $f^j$ for $1\le j<n$ has one endpoint falls in the open interval in between $r_n^j$ and $\ell_n^{j+1}$. The other endpoints of type-$f$ chords are located counterclockwise immediately after $c_1^1$ and arranged from $j=1$ to $n-1$. This arrangement ensures that those type-$f$ chords do not intersect each other. The type-$f'$ chord $f'^j$ where $1\le j<n$ is exactly and only intersects the type-$f$ chord $f^j$ with the same index $j$. 
Fig.~\ref{figure:circle_reduction_part}$(b)$ shows the types of chords mentioned above, specifically emphasizing the type-$e$ and type-$f$ chords with superscript $j=1$.

The type-$b$ chord $b_i^j$, where $1\le i,j\le n$, has one endpoint positioned counterclockwise immediately after $f'^1$, arranged in counterclockwise order from $i=1$ to $n$ and from $j=1$ to $n$. The other endpoint of each $b_i^j$ is positioned at a specific intersection with $c_i^j$ while avoiding $\ell_i^j$ and $r_i^j$. The type-$a$ chord $a_i$ for $1\le i\le n$ exclusively intersects all $b_i^j$ with the same index $i$. The type-$b'$ chord $b'^j_i$ for $1\le i,j\le n$ and type-$a'$ chord $a'_i$ for $1\le i\le n$ merely intersects with $b^j_i$ and $a_i$, respectively. In Fig.~\ref{figure:circle_reduction_part}$(c)$, all type of chords are depicted, highlighting \rgb{those} type-$b$ and type-$a$ chords associated with the subscript $j=1$.

The cardinality of $J$ is within a polynomial bound: there are $n^2-1$ type-$\ell$ and type-$r$ chords, $n^2$ type-$c$ chords, $nm$ type-$e$ chords, $n-1$ type-$f$ and type-$f'$ chords, $n^2$ type-$b$ and type-$b'$ chords, and $n$ type-$a$ and type-$a'$ chords. Altogether, there are $5n^2+4n-4+nm$ chords for $J$. Therefore, both the circle graph $G_c$ and the chord intersection model $J$ can be constructed in $O(n \cdot (n+m))$ time.

\begin{figure}
	\begin{center}
		\begin{tikzpicture}[scale=0.57]
\def\r {5.5}
\def\rr {11}
\def\k {16}
\def\deg {360/\k}
\def\bias {0.2}


		\draw (0,0) circle (\r);
		
		\foreach \i in {8,...,16}
		{
		        \draw[thick]({\r*cos((\i-\bias)*\deg)},{\r*sin((\i-\bias)*\deg)})..controls ({0.9*\r*cos((\i-0.25*\bias)*\deg)},{0.9*\r*sin((\i-0.25*\bias)*\deg)}) and ({0.9*\r*cos((\i+0.25*\bias)*\deg)},{0.9*\r*sin((\i+0.25*\bias)*\deg)})..({\r*cos((\i+\bias)*\deg)},{\r*sin((\i+\bias)*\deg)});
		}
		
		\foreach \i in {1,2,3}
		{
			\foreach \j in {1,2,3}
			{
				{\node at ({(\r+0.5)*cos((\i+3*\j+4)*\deg)},{(\r+0.5)*sin((\i+3*\j+4)*\deg)}){$c_{\i}^{\j}$};}
				
				\ifthenelse{\NOT \i=1 \OR \NOT \j=1}{\node at ({(\r+0.5)*cos((\i+3*\j+4-1.75*\bias)*\deg)},{(\r+0.5)*sin((\i+3*\j+4-1.75*\bias)*\deg)}){$\ell_{\i}^{\j}$};}
				
				\ifthenelse{\NOT \i=3 \OR \NOT \j=3}{\node at ({(\r+0.5)*cos((\i+3*\j+4+1.75*\bias)*\deg)},{(\r+0.5)*sin((\i+3*\j+4+1.75*\bias)*\deg)}){$r_{\i}^{\j}$};}
			}
			
		}

		\foreach \i in {8,...,15}
		{
		        \draw[thick]({\r*cos((\i+0.25*\bias)*\deg)},{\r*sin((\i+0.25*\bias)*\deg)})..controls ({0.9*\r*cos((\i+\bias-0.25*\bias)*\deg)},{0.9*\r*sin((\i+\bias-0.25*\bias)*\deg)}) and ({0.9*\r*cos((\i+\bias+0.25*\bias)*\deg)},{0.9*\r*sin((\i+\bias+0.25*\bias)*\deg)})..({\r*cos((\i+1.75*\bias)*\deg)},{\r*sin((\i+1.75*\bias)*\deg)});
		}
		
		\foreach \i in {9,...,16}
		{
		        \draw[thick]({\r*cos((\i-1.75*\bias)*\deg)},{\r*sin((\i-1.75*\bias)*\deg)})..controls ({0.9*\r*cos((\i-\bias-0.25*\bias)*\deg)},{0.9*\r*sin((\i-\bias-0.25*\bias)*\deg)}) and ({0.9*\r*cos((\i-\bias+0.25*\bias)*\deg)},{0.9*\r*sin((\i-\bias+0.25*\bias)*\deg)})..({\r*cos((\i-0.25*\bias)*\deg)},{\r*sin((\i-0.25*\bias)*\deg)});
		}
		
		\node at (0,-1.25*\r){$(a)$};

		
		\end{tikzpicture}
            \begin{tikzpicture}[scale=0.57]
\def\r {5.5}
\def\rr {11}
\def\k {16}
\def\deg {360/\k}
\def\bias {0.2}


		\draw (0,0) circle (\r);
		
		\foreach \i in {8,...,16}
		{
		        \draw[gray]({\r*cos((\i-\bias)*\deg)},{\r*sin((\i-\bias)*\deg)})..controls ({0.9*\r*cos((\i-0.25*\bias)*\deg)},{0.9*\r*sin((\i-0.25*\bias)*\deg)}) and ({0.9*\r*cos((\i+0.25*\bias)*\deg)},{0.9*\r*sin((\i+0.25*\bias)*\deg)})..({\r*cos((\i+\bias)*\deg)},{\r*sin((\i+\bias)*\deg)});
		}
		
		\foreach \i in {1,2,3}
		{
			\foreach \j in {1,2,3}
			{
				{\node [gray] at ({(\r+0.5)*cos((\i+3*\j+4)*\deg)},{(\r+0.5)*sin((\i+3*\j+4)*\deg)}){$c_{\i}^{\j}$};}
				
				\ifthenelse{\NOT \i=1 \OR \NOT \j=1}{\node [gray] at ({(\r+0.5)*cos((\i+3*\j+4-1.75*\bias)*\deg)},{(\r+0.5)*sin((\i+3*\j+4-1.75*\bias)*\deg)}){$\ell_{\i}^{\j}$};}
				
				\ifthenelse{\NOT \i=3 \OR \NOT \j=3}{\node [gray] at ({(\r+0.5)*cos((\i+3*\j+4+1.75*\bias)*\deg)},{(\r+0.5)*sin((\i+3*\j+4+1.75*\bias)*\deg)}){$r_{\i}^{\j}$};}
			}
		}

		\foreach \i in {8,...,15}
		{
		        \draw[gray]({\r*cos((\i+0.25*\bias)*\deg)},{\r*sin((\i+0.25*\bias)*\deg)})..controls ({0.9*\r*cos((\i+\bias-0.25*\bias)*\deg)},{0.9*\r*sin((\i+\bias-0.25*\bias)*\deg)}) and ({0.9*\r*cos((\i+\bias+0.25*\bias)*\deg)},{0.9*\r*sin((\i+\bias+0.25*\bias)*\deg)})..({\r*cos((\i+1.75*\bias)*\deg)},{\r*sin((\i+1.75*\bias)*\deg)});
		}
		
		\foreach \i in {9,...,16}
		{
		        \draw[gray]({\r*cos((\i-1.75*\bias)*\deg)},{\r*sin((\i-1.75*\bias)*\deg)})..controls ({0.9*\r*cos((\i-\bias-0.25*\bias)*\deg)},{0.9*\r*sin((\i-\bias-0.25*\bias)*\deg)}) and ({0.9*\r*cos((\i-\bias+0.25*\bias)*\deg)},{0.9*\r*sin((\i-\bias+0.25*\bias)*\deg)})..({\r*cos((\i-0.25*\bias)*\deg)},{\r*sin((\i-0.25*\bias)*\deg)});
		}
		
		\foreach \i in {1,4}
		{
		        \draw[gray]({\r*cos(((\i+8)+1.5*\bias)*\deg)},{\r*sin(((\i+8)+1.5*\bias)*\deg)})..controls ({0.4*\r*cos(((\i+10)-1.25*\bias)*\deg)},{0.4*\r*sin(((\i+10)-1.25*\bias)*\deg)}) and ({0.4*\r*cos(((\i+10)+1.5*\bias)*\deg)},{0.4*\r*sin(((\i+10)+1.5*\bias)*\deg)})..({\r*cos(((\i+12)-1.25*\bias)*\deg)},{\r*sin(((\i+12)-1.25*\bias)*\deg)});
		}
		
		\foreach \i in {1,4}
		{
		        \draw[thick]({\r*cos(((\i+7)+1.5*\bias)*\deg)},{\r*sin(((\i+7)+1.5*\bias)*\deg)})..controls({0.3*\r*cos(((\i+9.5)-1.25*\bias)*\deg)},{0.3*\r*sin(((\i+9.5)-1.25*\bias)*\deg)}) and ({0.3*\r*cos(((\i+9.5)+1.5*\bias)*\deg)},{0.3*\r*sin(((\i+9.5)+1.5*\bias)*\deg)})..({\r*cos(((\i+12)-1.5*\bias)*\deg)},{\r*sin(((\i+12)-1.5*\bias)*\deg)});

		}
		
		\node [gray] at (-7.2*\r/12,-4.45*\r/12){$e_{2,3}^1$};
		\node [gray] at (1.54*\r/12,-6.8*\r/12){$e_{2,3}^2$};
		\node at (-7.2*\r/12,-1.4*\r/12){$e_{1,3}^1$};
		\node at (-1.95*\r/12,-7.27*\r/12){$e_{1,3}^2$};

		\node at (-1.2*\r,0){};
		\node at (0,-1.4*\r){$(b)$};

		
		\end{tikzpicture}
		\begin{tikzpicture}[scale=0.57]
\def\r {5.5}
\def\rr {11}
\def\k {16}
\def\deg {360/\k}
\def\bias {0.2}


		\draw (0,0) circle (\r);
		
		\foreach \i in {8,...,16}
		{
		        \draw[gray]({\r*cos((\i-\bias)*\deg)},{\r*sin((\i-\bias)*\deg)})..controls ({0.9*\r*cos((\i-0.25*\bias)*\deg)},{0.9*\r*sin((\i-0.25*\bias)*\deg)}) and ({0.9*\r*cos((\i+0.25*\bias)*\deg)},{0.9*\r*sin((\i+0.25*\bias)*\deg)})..({\r*cos((\i+\bias)*\deg)},{\r*sin((\i+\bias)*\deg)});
		}
		
		\foreach \i in {1,2,3}
		{
			\foreach \j in {1,2,3}
			{
				\node [gray] at ({(\r+0.5)*cos((\i+3*\j+4)*\deg)},{(\r+0.5)*sin((\i+3*\j+4)*\deg)}){$c_{\i}^{\j}$};
				
				\ifthenelse{\NOT \i=1 \OR \NOT \j=1}{\node [gray] at ({(\r+0.5)*cos((\i+3*\j+4-1.75*\bias)*\deg)},{(\r+0.5)*sin((\i+3*\j+4-1.75*\bias)*\deg)}){$\ell^{\j}_{\i}$};}
				
				\ifthenelse{\NOT \i=3 \OR \NOT \j=3}{\node [gray] at ({(\r+0.5)*cos((\i+3*\j+4+1.75*\bias)*\deg)},{(\r+0.5)*sin((\i+3*\j+4+1.75*\bias)*\deg)}){$r^{\j}_{\i}$};}
			}
        }

		\foreach \i in {8,...,15}
		{
		        \draw[gray]({\r*cos((\i+0.25*\bias)*\deg)},{\r*sin((\i+0.25*\bias)*\deg)})..controls ({0.9*\r*cos((\i+\bias-0.25*\bias)*\deg)},{0.9*\r*sin((\i+\bias-0.25*\bias)*\deg)}) and ({0.9*\r*cos((\i+\bias+0.25*\bias)*\deg)},{0.9*\r*sin((\i+\bias+0.25*\bias)*\deg)})..({\r*cos((\i+1.75*\bias)*\deg)},{\r*sin((\i+1.75*\bias)*\deg)});
		}
		
		\foreach \i in {9,...,16}
		{
		        \draw[gray]({\r*cos((\i-1.75*\bias)*\deg)},{\r*sin((\i-1.75*\bias)*\deg)})..controls ({0.9*\r*cos((\i-\bias-0.25*\bias)*\deg)},{0.9*\r*sin((\i-\bias-0.25*\bias)*\deg)}) and ({0.9*\r*cos((\i-\bias+0.25*\bias)*\deg)},{0.9*\r*sin((\i-\bias+0.25*\bias)*\deg)})..({\r*cos((\i-0.25*\bias)*\deg)},{\r*sin((\i-0.25*\bias)*\deg)});
		}

        \foreach \i in {4}
		{
		        \draw[gray]({\r*cos(((\i+7)+1.5*\bias)*\deg)},{\r*sin(((\i+7)+1.5*\bias)*\deg)})..controls({0.3*\r*cos(((\i+9.5)-1.25*\bias)*\deg)},{0.3*\r*sin(((\i+9.5)-1.25*\bias)*\deg)}) and ({0.3*\r*cos(((\i+9.5)+1.5*\bias)*\deg)},{0.3*\r*sin(((\i+9.5)+1.5*\bias)*\deg)})..({\r*cos(((\i+12)-1.5*\bias)*\deg)},{\r*sin(((\i+12)-1.5*\bias)*\deg)});
		        
		        \draw[gray]({\r*cos(((\i+8)+1.5*\bias)*\deg)},{\r*sin(((\i+8)+1.5*\bias)*\deg)})..controls ({0.4*\r*cos(((\i+10)-1.25*\bias)*\deg)},{0.4*\r*sin(((\i+10)-1.25*\bias)*\deg)}) and ({0.4*\r*cos(((\i+10)+1.5*\bias)*\deg)},{0.4*\r*sin(((\i+10)+1.5*\bias)*\deg)})..({\r*cos(((\i+12)-1.25*\bias)*\deg)},{\r*sin(((\i+12)-1.25*\bias)*\deg)});
		}
		
		\foreach \i in {1}
		{
		        \draw[thick]({\r*cos(((\i+7)+1.5*\bias)*\deg)},{\r*sin(((\i+7)+1.5*\bias)*\deg)})..controls({0.3*\r*cos(((\i+9.5)-1.25*\bias)*\deg)},{0.3*\r*sin(((\i+9.5)-1.25*\bias)*\deg)}) and ({0.3*\r*cos(((\i+9.5)+1.5*\bias)*\deg)},{0.3*\r*sin(((\i+9.5)+1.5*\bias)*\deg)})..({\r*cos(((\i+12)-1.5*\bias)*\deg)},{\r*sin(((\i+12)-1.5*\bias)*\deg)});
		        
		        \draw[thick]({\r*cos(((\i+8)+1.5*\bias)*\deg)},{\r*sin(((\i+8)+1.5*\bias)*\deg)})..controls ({0.4*\r*cos(((\i+10)-1.25*\bias)*\deg)},{0.4*\r*sin(((\i+10)-1.25*\bias)*\deg)}) and ({0.4*\r*cos(((\i+10)+1.5*\bias)*\deg)},{0.4*\r*sin(((\i+10)+1.5*\bias)*\deg)})..({\r*cos(((\i+12)-1.25*\bias)*\deg)},{\r*sin(((\i+12)-1.25*\bias)*\deg)});
		}

		\node at (-7.2*\r/12,-1.3*\r/12){$e_{1,3}^1$};
		\node [gray] at (-7.2*\r/12,-4.3*\r/12){$e_{2,3}^1$};
		\node at (-1.65*\r/12,-6.75*\r/12){$e_{1,3}^2$};
		\node [gray] at (1.54*\r/12,-6.7*\r/12){$e_{2,3}^2$};

		\foreach \i in {6.75,7}
		{
			\draw[gray] ({\r*cos((\i-0.5*\bias)*\deg)},{\r*sin((\i-0.5*\bias)*\deg)})..controls({0.9*\r*cos((\i-0.25*\bias)*\deg)},{0.9*\r*sin((\i-0.25*\bias)*\deg)}) and ({0.9*\r*cos((\i+0.25*\bias)*\deg)},{0.9*\r*sin((\i+0.25*\bias)*\deg)})..({\r*cos((\i+0.5*\bias)*\deg)},{\r*sin((\i+0.5*\bias)*\deg)});
		}
  
		\draw[gray]({\r*cos(7*\deg)},{\r*sin(7*\deg)})..controls (0,0)..({\r*cos((13.5)*\deg)},{\r*sin((13.5)*\deg)});
		\draw[thick]({\r*cos(6.75*\deg)},{\r*sin(6.75*\deg)})..controls (0,0)..({\r*cos((10.5)*\deg)},{\r*sin((10.5)*\deg)});
		\node [gray] at ({(\r+0.5)*cos(7*\deg)},{(\r+0.5)*sin(7*\deg)}){$f^2$};
		\node at ({(\r+0.5)*cos(6.65*\deg)},{(\r+0.5)*sin(6.65*\deg)}){$f^1$};
		
		\node [gray] at ({(0.93*\r)*cos(7*\deg+\r)},{(0.93*\r)*sin(7*\deg+\r)}){$f'^2$};
		\node [gray] at ({(0.93*\r)*cos(6.65*\deg-\r)},{(0.93*\r)*sin(6.65*\deg-\r)}){$f'^1$};
		
		\node at (0,1.025*\r){};
		\node at (0,-1.25*\r){$(c)$};
		\end{tikzpicture}
		\begin{tikzpicture}[scale=0.57]
\def\r {5.5}
\def\rr {11}
\def\k {16}
\def\deg {360/\k}
\def\bias {0.2}


		\draw (0,0) circle (\r);
		
		\foreach \i in {8,...,16}
		{
		        \draw[gray]({\r*cos((\i-\bias)*\deg)},{\r*sin((\i-\bias)*\deg)})..controls ({0.9*\r*cos((\i-0.25*\bias)*\deg)},{0.9*\r*sin((\i-0.25*\bias)*\deg)}) and ({0.9*\r*cos((\i+0.25*\bias)*\deg)},{0.9*\r*sin((\i+0.25*\bias)*\deg)})..({\r*cos((\i+\bias)*\deg)},{\r*sin((\i+\bias)*\deg)});
		}
		
		\foreach \i in {1,2,3}
		{
			\foreach \j in {1,2,3}
			{
				{\node [gray] at ({(\r+0.5)*cos((\i+3*\j+4)*\deg)},{(\r+0.5)*sin((\i+3*\j+4)*\deg)}){$c_{\i}^{\j}$};}
				
				\ifthenelse{\NOT \i=1 \OR \NOT \j=1}{\node [gray] at ({(\r+0.5)*cos((\i+3*\j+4-1.75*\bias)*\deg)},{(\r+0.5)*sin((\i+3*\j+4-1.75*\bias)*\deg)}){$\ell_{\i}^{\j}$};}
				
				\ifthenelse{\NOT \i=3 \OR \NOT \j=3}{\node [gray] at ({(\r+0.5)*cos((\i+3*\j+4+1.75*\bias)*\deg)},{(\r+0.5)*sin((\i+3*\j+4+1.75*\bias)*\deg)}){$r_{\i}^{\j}$};}
			}
			
		}

		\foreach \i in {8,...,15}
		{
		        \draw[gray]({\r*cos((\i+0.25*\bias)*\deg)},{\r*sin((\i+0.25*\bias)*\deg)})..controls ({0.9*\r*cos((\i+\bias-0.25*\bias)*\deg)},{0.9*\r*sin((\i+\bias-0.25*\bias)*\deg)}) and ({0.9*\r*cos((\i+\bias+0.25*\bias)*\deg)},{0.9*\r*sin((\i+\bias+0.25*\bias)*\deg)})..({\r*cos((\i+1.75*\bias)*\deg)},{\r*sin((\i+1.75*\bias)*\deg)});
		}
		
		\foreach \i in {9,...,16}
		{
		        \draw[gray]({\r*cos((\i-1.75*\bias)*\deg)},{\r*sin((\i-1.75*\bias)*\deg)})..controls ({0.9*\r*cos((\i-\bias-0.25*\bias)*\deg)},{0.9*\r*sin((\i-\bias-0.25*\bias)*\deg)}) and ({0.9*\r*cos((\i-\bias+0.25*\bias)*\deg)},{0.9*\r*sin((\i-\bias+0.25*\bias)*\deg)})..({\r*cos((\i-0.25*\bias)*\deg)},{\r*sin((\i-0.25*\bias)*\deg)});
		}
		
		\foreach \i in {1,4}
		{
		        \draw[gray]({\r*cos(((\i+7)+1.5*\bias)*\deg)},{\r*sin(((\i+7)+1.5*\bias)*\deg)})..controls({0.3*\r*cos(((\i+9.5)-1.25*\bias)*\deg)},{0.3*\r*sin(((\i+9.5)-1.25*\bias)*\deg)}) and ({0.3*\r*cos(((\i+9.5)+1.5*\bias)*\deg)},{0.3*\r*sin(((\i+9.5)+1.5*\bias)*\deg)})..({\r*cos(((\i+12)-1.5*\bias)*\deg)},{\r*sin(((\i+12)-1.5*\bias)*\deg)});
		        
		        \draw[gray]({\r*cos(((\i+8)+1.5*\bias)*\deg)},{\r*sin(((\i+8)+1.5*\bias)*\deg)})..controls ({0.4*\r*cos(((\i+10)-1.25*\bias)*\deg)},{0.4*\r*sin(((\i+10)-1.25*\bias)*\deg)}) and ({0.4*\r*cos(((\i+10)+1.5*\bias)*\deg)},{0.4*\r*sin(((\i+10)+1.5*\bias)*\deg)})..({\r*cos(((\i+12)-1.25*\bias)*\deg)},{\r*sin(((\i+12)-1.25*\bias)*\deg)});
		}

		\node [gray] at (-7.2*\r/12,-1.3*\r/12){$e_{1,3}^1$};
		\node [gray] at (-7.2*\r/12,-4.35*\r/12){$e_{2,3}^1$};
		\node [gray] at (-1.65*\r/12,-6.75*\r/12){$e_{1,3}^2$};
		\node [gray] at (1.54*\r/12,-6.7*\r/12){$e_{2,3}^2$};

		\draw[gray]({\r*cos(6.75*\deg)},{\r*sin(6.75*\deg)})..controls (0,0)..({\r*cos((10.5)*\deg)},{\r*sin((10.5)*\deg)});
		\draw[gray]({\r*cos(7*\deg)},{\r*sin(7*\deg)})..controls (0,0)..({\r*cos((13.5)*\deg)},{\r*sin((13.5)*\deg)});
		\node [gray] at ({(\r+0.5)*cos(7*\deg)},{(\r+0.5)*sin(7*\deg)}){$f^2$};
		\node [gray] at ({(\r+0.5)*cos(6.65*\deg)},{(\r+0.5)*sin(6.65*\deg)}){$f^1$};

		\node [gray] at ({(0.93*\r)*cos(7*\deg+\r)},{(0.93*\r)*sin(7*\deg+\r)}){$f'^2$};
		\node [gray] at ({(0.93*\r)*cos(6.65*\deg-\r)},{(0.93*\r)*sin(6.65*\deg-\r)}){$f'^1$};
		
		\foreach \i in {6.75,7}
		{
			\draw[gray]({\r*cos((\i-0.5*\bias)*\deg)},{\r*sin((\i-0.5*\bias)*\deg)})..controls({0.9*\r*cos((\i-0.25*\bias)*\deg)},{0.9*\r*sin((\i-0.25*\bias)*\deg)}) and ({0.9*\r*cos((\i+0.25*\bias)*\deg)},{0.9*\r*sin((\i+0.25*\bias)*\deg)})..({\r*cos((\i+0.5*\bias)*\deg)},{\r*sin((\i+0.5*\bias)*\deg)});
		}

		\foreach \i in {2,4}
		{
			\draw[gray]({\r*cos((1.125+\i)*\deg)},{\r*sin((1.125+\i)*\deg)})..controls (0,0)..({\r*cos((8+\i/2)*\deg)},{\r*sin((8+\i/2)*\deg)});
			\draw[gray]({\r*cos((1.5+\i)*\deg)},{\r*sin((1.5+\i)*\deg)})..controls (0,0)..({\r*cos((11+\i/2)*\deg)},{\r*sin((11+\i/2)*\deg)});
			\draw[gray]({\r*cos((1.875+\i)*\deg)},{\r*sin((1.875+\i)*\deg)})..controls (0,0)..({\r*cos((14+\i/2)*\deg)},{\r*sin((14+\i/2)*\deg)});
		}
		
		\foreach \i in {2.6,3.125,3.5,3.875,4.6,5.125,5.5,5.875}
		{
			\draw[gray]({\r*cos((\i-0.5*\bias)*\deg)},{\r*sin((\i-0.5*\bias)*\deg)})..controls({0.9*\r*cos((\i-0.25*\bias)*\deg)},{0.9*\r*sin((\i-0.25*\bias)*\deg)}) and ({0.9*\r*cos((\i+0.25*\bias)*\deg)},{0.9*\r*sin((\i+0.25*\bias)*\deg)})..({\r*cos((\i+0.5*\bias)*\deg)},{\r*sin((\i+0.5*\bias)*\deg)});
		}
		
		\foreach \i in {3,5}
		{
			\draw[gray]({\r*cos((\i-0.4)*\deg)},{\r*sin((\i-0.4)*\deg)})..controls({0.6*\r*cos((\i-0.05)*\deg)},{0.6*\r*sin((\i-0.05)*\deg)}) and ({0.6*\r*cos((\i+0.75)*\deg)},{0.6*\r*sin((\i+0.75)*\deg)})..({\r*cos((\i+1.1)*\deg)},{\r*sin((\i+1.1)*\deg)});
		}
		
		\foreach \i in {0.6,1.125,1.5,1.875}
		{
			\draw[gray]({\r*cos((\i-0.5*\bias)*\deg)},{\r*sin((\i-0.5*\bias)*\deg)})..controls({0.9*\r*cos((\i-0.25*\bias)*\deg)},{0.9*\r*sin((\i-0.25*\bias)*\deg)}) and ({0.9*\r*cos((\i+0.25*\bias)*\deg)},{0.9*\r*sin((\i+0.25*\bias)*\deg)})..({\r*cos((\i+0.5*\bias)*\deg)},{\r*sin((\i+0.5*\bias)*\deg)});
		}
		
		\foreach \i in {0}
		{
			\draw[thick]({\r*cos((1.125+\i)*\deg)},{\r*sin((1.125+\i)*\deg)})..controls (0,0)..({\r*cos((8+\i/2)*\deg)},{\r*sin((8+\i/2)*\deg)});
			\draw[thick]({\r*cos((1.5+\i)*\deg)},{\r*sin((1.5+\i)*\deg)})..controls (0,0)..({\r*cos((11+\i/2)*\deg)},{\r*sin((11+\i/2)*\deg)});
			\draw[thick]({\r*cos((1.875+\i)*\deg)},{\r*sin((1.875+\i)*\deg)})..controls (0,0)..({\r*cos((14+\i/2)*\deg)},{\r*sin((14+\i/2)*\deg)});
		}

		\foreach \i in {1}
		{
			\draw[thick]({\r*cos((\i-0.4)*\deg)},{\r*sin((\i-0.4)*\deg)})..controls({0.6*\r*cos((\i-0.05)*\deg)},{0.6*\r*sin((\i-0.05)*\deg)}) and ({0.6*\r*cos((\i+0.75)*\deg)},{0.6*\r*sin((\i+0.75)*\deg)})..({\r*cos((\i+1.1)*\deg)},{\r*sin((\i+1.1)*\deg)});
		}

		\foreach \i in {1}
		{
			\node at ({(\r+0.5)*cos((\i-0.4)*\deg)},{(\r+0.5)*sin((\i-0.4)*\deg)}){$a_{\i}$};
			\node [gray] at ({(0.9*\r)*cos((2*\i-1.62)*\deg)},{(0.9*\r)*sin((2*\i-1.62)*\deg)}){$a'_{\i}$};
			\foreach \j in {1,2,3}
			{
				\node at ({(\r+0.5)*cos((2*\i+3*\j/8-1.25)*\deg)},{(\r+0.5)*sin((2*\i+3*\j/8-1.25)*\deg)}){$b_{\i}^{\j}$};
				\node [gray] at ({(0.88*\r)*cos((2*\i+3*\j/8-1.25)*\deg-0.8*\r)},{(0.88*\r)*sin((2*\i+3*\j/8-1.25)*\deg-0.8*\r)}){$b'^{\j}_{\i}$};
			}
		}
		
		\foreach \i in {2,3}
		{
			\node [gray] at ({(\r+0.5)*cos((2*\i-1.4)*\deg)},{(\r+0.5)*sin((2*\i-1.4)*\deg)}){$a_{\i}$};
			\node [gray] at ({(0.9*\r)*cos((2*\i-1.6)*\deg)},{(0.9*\r)*sin((2*\i-1.6)*\deg)}){$a'_{\i}$};
			\foreach \j in {1,2,3}
			{
				\node [gray] at ({(\r+0.5)*cos((2*\i+3*\j/8-1.25)*\deg)},{(\r+0.5)*sin((2*\i+3*\j/8-1.25)*\deg)}){$b_{\i}^{\j}$};
				\node [gray] at ({(0.88*\r)*cos((2*\i+3*\j/8-1.25)*\deg-0.8*\r)},{(0.88*\r)*sin((2*\i+3*\j/8-1.25)*\deg-0.8*\r)}){$b'^{\j}_{\i}$};
			}
		}

	\node at (0,-1.25*\r){$(d)$};
		\end{tikzpicture}
	\end{center}
\vspace{-10pt}
\caption{The corresponding intersection model $J$ of the given instance $G_H=\{\{v_1,v_2,v_3\},\{(v_1,v_3),(v_2,v_3)\}\}$.}
\label{figure:circle_reduction_part}
\end{figure}

Now, we demonstrate the correctness of Theorem~\ref{theorem:Hamiltonian_to_MPDS}, by showing $G_H$ contains a Hamiltonian path if and only if $J$ has a paired-dominating set with cardinality $2n^2+2n-2$, as stated in Lemma~\ref{lemma:Hamiltonian_->_MPDS} and Lemma~\ref{lemma:Hamiltonian_<-_MPDS} for the necessity and sufficient conditions, respectively. 


\begin{Lemma}\label{lemma:Hamiltonian_->_MPDS}
	If $G_H$ contains a Hamiltonian path , then $J$ has a paired-dominating set with cardinality $2n^2+2n-2$.
\end{Lemma}
\begin{proof}
    Suppose that $G_H$ contains a Hamiltonian Path from $v_1$ to $v_n$. There exists an invertible function $\sigma:\{1,2,\ldots,n\}\rightarrow\{1,2,\ldots,n\}$, where ${\sigma(1)}=1$ and ${\sigma(n)}=n$, such that the sequence $P=(v_{\sigma(1)}, v_{\sigma(2)}, \ldots, v_{\sigma(n)})$ forms a Hamiltonian path. Note that the $i$-th node in $P$ is represented by $\sigma^{-1}(i)$.
    Let $S$ be the chord set consists of all type-$b$ and type-$f$ chords initially. One can find that $|S|=n^2+n-1$ and $N_J[S]$ including all type-$b$, type-$b'$, type-$f$, and type-$e$ chords. Thus, it suffices to transform $S$ into a paired-dominating of $J$ by subsequently adding a unique neighbor for each chord $c\in S$ such that finally the remaining type-$a$, type-$a'$, type-$\ell$, and type-$r$ chords are all in $N_J[S]$.

    For each $b^{i}_{\sigma(i)}$ where $1\le i\le n$, we add $a^i\in N_J[b^{i}_{\sigma(i)}]$ into $S$. This ensures that all type-$a$ and type-$a'$ chords are included by $N_J[S]$. Furthermore, for each $f^j$ where $1\le j<n$, we add $e^j_{\sigma^{-1}(j),\sigma^{-1}(j+1)}\in N_J[f^j]$ into $S$ to make $\ell^j_{\sigma^{-1}(j)}$ and $r^{j+1}_{\sigma^{-1}(j+1)}$ are included in $N_J[S]$. One can verify that the remaining chords needing to be dominated are $\ell_i^j$ and $r_i^j$ with $j\neq\sigma(i)$. These $i,j$ correspond precisely to the chords $b_i^j$ in $S$ for which we have not yet found a neighbor. We can add all $c_i^j\in N_J[b_i^j]$ into $S$ and finally make every type-$\ell$ and type-$r$ chord in $N_J[S]$.
\end{proof}

\begin{Lemma}\label{lemma:Hamiltonian_<-_MPDS}
	If $G_c$ has a paired-dominating set with cardinality $2n^2+2n-2$, then $G_H$ contains a Hamiltonian path.
\end{Lemma}
\begin{proof}
    Suppose $D$ is a paired-dominating set of $J$ with cardinality $2n^2+2n-2$. Since each type-$a'$, type-$b'$, and type-$f'$ chord is in $N[D]$, $D$ \rgb{includes} the union of all type-$a$, type-$b$, and type-$f$ chords, which we denoted by $X$. One can verify that $|X|=n^2+2n-1$ and $N[X]$ contains no any type-$\ell$ or type-$r$ chord. Furthermore, the cardinality of type-$\ell$ and type-$r$ chords in total is $2n^2-2$, which equals $2|D-X|$. Those imply that $D-X$ \rgb{consists} of some type-$c$ and type-$e$ chords, as those are the only candidates that can intersect at most twice with the type-$\ell$ and type-$r$ chords.

    Note that the candidates of $a_i$ to be paired in a perfect matching of $J[D]$ are only the type-$b$ chords with the same subscript $i$ for $1\le i\le n$, which we denoted those type-$b$ chord as $b_i^{p(i)}$. Thus, it suffices to show that $p^{-1}(i)$ is unique for $1\le i\le n$ while $e_{p^{-1}(i-1),p^{-1}(i)}^i\in D$ if $i>1$, proving the existence of a Hamiltonian path $P=(p^{-1}(1),p^{-1}(2),\ldots,p^{-1}(n))$. y using the induction based on the index $i$. Considering the initial value $i=1$, for $j\neq 1$, $c_1^j$ is the only chord intersecting $\ell_1^j$ and twice type-$\ell$ and type-$r$ chords which have not been dominated yet, implying that $c_1^j\in D$ paired with $b_1^j$ for $j\neq 1$. Thus, $p^{-1}(1)=1$ is unique. y applying the induction hypothesis, using the similar method, one can refer that $c_{i-1}^j\in D$ for $j\neq p^{-1}(i-1)$. Therefore, the candidates for intersecting $\ell_i^j$ and twice the chords of type-$\ell$ and type-$r$ which have not been dominated yet should be $c_i^j$ or $e_{p^{-1}(i-1),i}^{j-1}$. However, $b_i^{p^{(-1)(i)}}$ have been paired with $a^{p^(-1)(i)}$, meaning $c_i^{p^{-1}(i)}\notin D$. Those imply that $e_{p^{-1}(i-1),p^{-1}(i)}^{j-1}\in D$ and $p^{-1}(i)$ is unique since $c_i^j\in D$ for $j\neq p^{-1}(i)$. Hence, the induction stands affirmed and the Hamiltonian path $P$ exists.
\end{proof}

Below, we show the NP-completeness of the paired-domination problem on circle graphs. Obviously, it falls in NP. Moreover, for each graph $G_H$, one can construct a circle graph $G_c$ and determine whether $G_H$ contains a Hamiltonian path or not by Theorem~\ref{theorem:Hamiltonian_to_MPDS}. Since the construction of $G_c$ \rgb{takes} polynomial time, the paired-domination problem on circle graph also remains NP-complete.

\begin{theorem}\label{theorem:circle_is_NPC}
	The paired-domination problem on circle graphs is NP-complete.
\end{theorem}

\smallskip
\section{Algorithm for $k$-polygon graphs}
\label{section:k-polygon}
In this section, we proposed two polynomial time algorithms, improved from work of~\cite{ElMallah93}, one for finding minimum dominating sets in $O(n(\frac{n}{k^2-k})^{2k^2-4k})$ time and the other for finding minimum paired-dominating sets in $O(n(\frac{n}{k^2-k})^{2k^2-2k})$ time, both on k-polygon graphs, a subclass of circle graphs. Next, we further improve the previous result for domination problem into a time complexity $O(n^{3k-5})$ and extend it to total domination problem in also $O(n^{3k-5})$.

Our results improved the work of~\cite{ElMallah93}. In the next two sections, the NP-hardness proof and the two proposed algorithms are presented. Finally, before ending this chapter, some remarks are given in Section~\ref{section:conclusion}.

\subsection{Domination between Two Sides}
\label{sectionalgo_for_two_sides}
In this subsection, we focus on the chords $c$ that lie on exactly two sides, $S_i$ and $S_j$, that is, $c\in J_{ij}$. We first consider the case where sides $S_i$ and $S_j$ are adjacent, assuming without loss of generality that $j = i+1$.
Given that $a_i\in S_i$ and $a_j\in S_j$, the following lemmas serve as auxiliary results for our algorithms addressing the domination and paired-domination problems of $J$.

\begin{lemma} \label{lemma:adjacent_O(n)_permutation}
    For $i=j-1$, the minimum set $S\subseteq J_{ij}$ satisfying $J'_{ij} = c([a_i, r_i(J)]\cap [\ell_j(J), a_j]) \subseteq N_{J}[S]$ can be determined in $O(|J_{ij}|)$ time.
\end{lemma}
\begin{proof}
    We union an auxiliary chord set $Aux = \{x_k \mid 1\le k\le 5\}$ with $J_{ij}$, assuming without loss of generality that each chord $x_i$ also intersects the chords of $J_{ij}$ sharing the same endpoints. The chord $x_1$ has one endpoint at $a_i-1$ on $S_i$ and the other at $r_j(J)+3$ on $S_j$, while the chord $x_2$ has one endpoint at $a_j+1$ on $S_j$ and the other at $\ell_i(J)-3$ on $S_i$. The chord $x_3$, which intersects only $x_1$ and $x_2$, has endpoints $\ell_i(J)-1$ on $S_i$ and $r_j(J)+1$ on $S_j$. The chord $x_4$, intersecting $x_1$, has one endpoint at $\ell_i(J)-2$ on $S_i$ and the other at $r_j(J)+4$ on $S_j$. Symmetrically, the chord $x_5$, intersecting $x_2$, has one endpoint at $r_j(J)+2$ on $S_j$ and the other at $\ell_i(J)+4$ on $S_i$

    Since $J_{ij}$ remains a chord set of $2$-polygon graph, which is equivalent to a permutation graph, one can apply the $O(n)$-time algorithm~\cite{ElMallah93} for finding a minimum dominating set $D$ of $J_{ij}$, where $n = |J_{ij}|$ in this context. Furthermore, given that $N[c(x_3)] = \{c(x_1),c(x_2),c(x_3)\}$, $N[c(x_4)] = \{c(x_1),c(x_4),c(x_5)\}$, and $N[c(x_5)] = \{c(x_2),c(x_4),c(x_5)\}$, it follows that there are at least two chord from $Aux$ in $D$. Moreover, since $Aux\subseteq N_{J_{ij}}[\{c(x_1),c(x_2)\}]$, we can suppose without loss of generality that $D\cap Aux = \{c(x_1),c(x_2)\}$. Additionally, one can verify that $J_{ij} - J'_{ij}\subseteq N_{J}[\{c(x_1),c(x_2)\}]$ and $J'_{ij}\cap N_{J}[\{c(x_1),c(x_2)\}] = \emptyset$. Therefore, the minimum set $S\subseteq J_{ij}$ satisfying $J'_{ij}\subseteq N_{J}[S]$ is $S= D-Aux$.
\end{proof}

\begin{lemma} \label{lemma:adjacent_paired_O(n)_permutation}
    For $i=j-1$, the minimum set $S\subseteq J_{ij}$ satisfying $J'_{ij} = c([a_i, r_i(J)]\cap [\ell_j(J), a_j]) \subseteq N_{J}[S]$, where $J[S]$ contains a perfect matching, can be determined in $O(|J_{ij}|)$ time.
\end{lemma}
\begin{proof}
    As in the proof of Lemma~\ref{lemma:adjacent_O(n)_permutation}, we take the union of the auxiliary chord set $Aux = \{x_k \mid 1\le k\le 5\}$ with $J_{ij}$. Since $J_{ij}$ is still a chord set of $2$-polygon graph, which is the same as a permutation graph, we can find a minimum paired-dominating set $D$ of $J_{ij}$, ensuring that $J_{ij}[D]$ contains a perfect matching $M$, by using the $O(n)$-time algorithm from~\cite{Lappas13}, where $n = |J_{ij}|$ here. Following a similar discussion as in the proof of Lemma~\ref{lemma:adjacent_O(n)_permutation}, we can assume without loss of generality that $D\cap Aux = \{c(x_1),c(x_2)\}$ and claim that the set $S = D - Aux$ is the minimum set satisfying $J'_{ij}\subseteq N_{J}[S]$. Hence, it suffices to show that $J[S]$ contains a perfect matching. If $(c(x_1),c(x_2))\in M$, the statement obviously holds. Otherwise, supposing that the $(c(x_1),d_1),(c(x_2),d_2)\in M$, where $d_1$ and $d_2$ are two specific chords of $D$. One can verify that $M \cup \{d_1,d_2\} - \{(c(x_1),d_1),(c(x_2),d_2)\}$ remains a perfect matching of $J[S]$. Thus, this lemma holds.
\end{proof}

We can extend the previous results to a more general version as below. Given that $b_i,c_i\in S_i$ and $b_j,c_j\in S_j$, the similar methods used to prove Lemmas~\ref{lemma:adjacent_O(n)_permutation} and~\ref{lemma:adjacent_paired_O(n)_permutation} can be similarly applied to establish the correctness of Lemmas~\ref{lemma:general_O(n)_permutation} and~\ref{lemma:general_paired_O(n)_permutation}, respectively.

\begin{lemma} \label{lemma:general_O(n)_permutation}
    The minimum set $S\subseteq J_{ij}$ satisfying $J'_{ij} = c([b_i, c_i]\cap [b_j, c_j]) \subseteq N_{J_{ij}}[S]$ can be determined in $O(|J_{ij}|)$ time.
\end{lemma}
    
\begin{lemma} \label{lemma:general_paired_O(n)_permutation}
    The minimum set $S\subseteq J_{ij}$ satisfying $J'_{ij} = c([b_i, c_i]\cap [b_j, c_j]) \subseteq N_{J_{ij}}[S]$, where $J[S]$ contains a perfect matching, can be determined in $O(|J_{ij}|)$ time.
\end{lemma}

\subsection{Algorithms for $k$-polygon graphs}
\label{sectionalgo_1_for_k-polygon}
In this subsection, we provide the initial algorithms for paired-domination and domination problems on $k$-polygon graphs. Given a chord set $D\subseteq J$, the definition of the outer boundary $\B$ of $D$ consists of $\B_{ij}$ as below. For $D$ as a dominating set of $J$, the chord set $J' = J-N_J[\B]$, denotes the subset of $J$ not dominated by the outer boundary $\B$.

\begin{center}  
    $
    \B_{ij}=
    \begin{cases}
        \emptyset   &\text{if $D_{ij}=\emptyset$},\\
	c(\{\ell_i(D_{ij}),r_j(D_{ij})\})  &\text{if $i=j-1$},\\
	c(\{\ell_t(D_{ij}), r_t(D_{ij})\mid t\in \{i,j\}\})     &\text{otherwise.}
    \end{cases}
    $
\end{center}

\begin{lemma} \label{lemma:outer_boundary}
    If $i=j-1$, there are endpoints $a_i\in S_i$ and $a_j\in S_j$ such that $J'_{ij} = c([a_i,a_j])$. Otherwise, there are endpoints $b_i,c_i\in S_i$ and $b_j,c_j\in S_j$ such that $J'_{ij} = c([b_i,c_i])\cap c([b_j,c_j])$.
    Furthermore, determining all $J'_{ij}$ can be done in $O(n)$ time, and for each chord $j\in J'_{ij}$, $j\notin N_J[D-D_{ij}]$.
\end{lemma}
\begin{proof}
    We first consider the case that $j = i+1$. Let $a_i = r_i(N_J[\B]) + 1$ and $a_j = \ell_i(N_J[\B]) - 1$, ensuring that $c([a_i,a_j])\subseteq J-N_J[\B]$. Furthermore, it can be verified that each chord in $J_{ij} - c([a_i,a_j])$ intersects either $c(r_i(\B))$ or $c(\ell_j(\B))$. Hence, $J'_{ij} = c([a_i,a_j])$. Then, consider the case that $j \neq i+1$. For simplicity, We first separate $\B$ into two part: $\B_1 = N_J[c\{[r_j(D_{ij}), \ell_i(D_{ij})]\}]$ and $\B_2 = N_J[c\{[r_i(D_{ij}), \ell_i(D_{ij})]\}]$.
    Similarly, let $b_i = r_i(\B_1) + 1$, $c_i = \ell_j(\B_2) - 1$, $b_j = \ell_j(\B_2) - 1$, and $c_j = \ell_j(\B_1) + 1$. 
    Those ensure that $c([b_i,c_i])\cap c([b_j,c_j])\subseteq J-N_J[\B]$. Moreover, one can verify that each chord in $J_{ij} - c([b_i,c_i])\cap c([b_j,c_j]))$ intersects at least one of $c(r_i(\B_1))$, $c(\ell_j(\B_1))$, $c(r_i(\B_2))$, and $c(\ell_j(\B_2))$. Thus, $J'_{ij} = c([b_i,c_i])\cap c([b_j,c_j])$.

    Since all such endpoints can be identified by scanning every endpoints of $J$ once, $J'_{ij}$ can be computed in $O(n)$ time for all $i,j$.
    Moreover, for each chord $j\in J'_{ij}$, because $j\notin N_J[\B]$, where $\B$ includes every chord with the leftmost and rightmost endpoints for each pair of sides of $D$, it follows that $j\notin N_J[D-D_{ij}]$.
\end{proof}

Considering $D$ as a paired-dominating set of $J$ with $J[D]$ containing a perfect matching $\mathcal M$, according to its outer boundary $\B$, we further define the inner boundary $\I$ as follows. For simplicity, denote those chords paired with another locating on the same two sides in $\mathcal M$ as $M_{ij}=\{d_1,d_2\mid d_1,d_2\in D_{ij}\text{ and }(d_1,d_2)\in \mathcal M\}$ and $D'_{ij}=(D-M_{ij})\cup B$.

\begin{center}  
    $
    \I_{ij}=
    \begin{cases}
        \emptyset   &\text{if $J'_{ij}=\emptyset$},\\
	c(\{r_i(D'_{ij}), \ell_j(D'_{ij})\})  &\text{if $i=j-1$},\\
	c(\{\ell_t(D'_{ij}), r_t(D'_{ij})\mid t\in \{i,j\}\})    &\text{otherwise.}
    \end{cases}
    $
\end{center}

Note that $J'' = J-N_J[\B\cup \I]$, denotes the chord subset of $J$ dominated by neither outer boundary $\B$ nor inner boundary $\I$. $J''_{ij}$ can be determined as described below. The proof is similar to that in Lemma~\ref{lemma:outer_boundary}, so we omit it here.

\begin{lemma} \label{lemma:inner_boundary}
    If $j = i+1$, there are endpoints $a_i\in S_i$ and $a_j\in S_j$ such that $J''_{ij} = c([a_i,a_j])$. Otherwise, there are endpoints $b_i,c_i\in S_i$ and $b_j,c_j\in S_j$ such that $J''_{ij} = c([b_i,c_i])\cap c([b_j,c_j])$. 
    Furthermore, determining all $J''_{ij}$ can be done in $O(n)$ time, and for each chord $j\in J''_{ij}$, $j\notin N_J[D-M_{ij}]$.
\end{lemma}
    

The idea of our algorithm is to exhaustively explore every possible outer boundary $\B$ and corresponding inner boundary $\I$. Then, based on these boundaries, we reconstruct them into a minimum paired-dominating set including them. Consequently, the minimum paired-dominating set of $G$ emerges.

\begin{algorithm}[htb]
	\caption{Determining the minimum paired-dominating set on $k$-polygon graphs}\label{algo:PDS_k-poly}
	\begin{algorithmic} [1]
		\baselineskip 14pt
		\REQUIRE a $k$-polygon graph $G$.
		\ENSURE  a minimum paired-dominating set of $G$.
		
		\STATE Let $J$ be the intersection model of $G$;
  
		\FOR{each possible $\B$ and corresponding $\I$}
		  \FOR{each pair of sides $S_i$ and $S_j$ of $J$}
                \STATE Determine a minimum set $D''_{ij}\subseteq J_{ij}$ such that $J''_{ij}\subseteq N_J[D''_{ij}]$ while $J[D''_{ij}]$ contains a perfect matching;
            \ENDFOR
            
            \IF{$N_J[D''\cup\B\cup\I] = J$}
                \STATE Determine a minimum chord set $\Psi\subseteq J$ such that $J[\B\cup\I\cup\Psi]$ contains a perfect matching;
                \STATE Let $D \leftarrow D''\cup\B\cup\I\cup \Psi$;
            \ENDIF
		\ENDFOR

		\RETURN the minimum set $D$ in Step~(8).
	\end{algorithmic}
\end{algorithm}

\begin{theorem} \label{theorem:PDS_k-poly}
	Algorithm~\ref{algo:PDS_k-poly} finds minimum paired-dominating sets of $G$ in $O(\frac{n^{2k^2-2k+1}}{(k^2-k)^{2k^2-2k}})$ time.
\end{theorem}
\begin{proof}
    The correctness of Algorithm~\ref{algo:PDS_k-poly} is ensured since at least one minimum paired-dominating set $D$, where $J[D]$ contains a perfect matching $\mathcal{M}$, with the outer boundary $\B$ and the inner boundary $\I$ has been explored. Furthermore, one can verify that $J_{ij}$ is exactly a permutation graph and, by Lemma~\ref{lemma:inner_boundary}, for each $j\in J''_{ij}$, $j\notin N_J[J-M_{ij}]$. Thus, each $J''_{ij}$ can be replaced by the set $D''_{ij}$, determining in Step~$(4)$, while still ensuring that $D$ remains a minimum paired-dominating set.

    The time complexity is shown below.
    We first discuss the number of every possible outer boundary $\B$ and corresponding inner boundary $\I$ in Step~$(2)$. Let $n_{ij}$ is the cardinality of $J_{ij}$. If $i=j-1$, the number of possibilities to exhaustively explore $\B_{ij}$ and $\I_{ij}$ is on the order of $O((n_{ij})^2)$. Otherwise, the number of possibilities to exhaustively explore $\B_{ij}$ and $\I_{ij}$ is on the order of $O((n_{ij})^4)$. However, for each $i'\in [i+1,j-1]$ and $j'\in [j+1,i-1]$, if $\B_{i'j'} \neq \emptyset$, then $J_{ij} - N_J[\B_{i'j'}]$ might not be empty, implying $\I_{ij} \neq \emptyset$. Otherwise, $J_{ij} \subseteq N_J[\B_{i'j'}]$, implying $\I_{ij} = \emptyset$, resulting in a contradiction. Thus, the bound of cardinality of $\B$ and $\I$ is leading by exhausting $\B_{ij}$ for each $i,j$ and $\I_{ij}$ for $i=j-1$. Therefore, there are $\prod_{1\le i<j\le k} {n_{ij}}^{4}$ possibilities of boundaries. By the AM-GM inequality $\sqrt[C(k,2)]{\prod_{1\le i<j\le k} n_{ij}}\le \frac{\sum\nolimits_{1\le i<j \le k} n_{ij}}{C(k,2)} = \frac{n}{C(k,2)}$, where $C(k,2) = \frac{k^2-k}{2}$ represents the combination number for $1\le i<j\le k$, the possibilities of boundaries is $O((\frac{n}{k^2-k})^{2k^2-2k})$.
    In Step~$(4)$, determining all $D''_{ij}$ can be done in totaling $O(n)$ time by Lemmas~\ref{lemma:adjacent_paired_O(n)_permutation} and~\ref{lemma:general_paired_O(n)_permutation}.
    In Step~$(7)$, one can use the $O(|V|^2|E|)$-time blossom algorithm~\cite{Edmonds65} for $J[\B\cup\I]$, which takes $O(k^8)$ time.
    Hence, the time complexity of Algorithm~\ref{algo:PDS_k-poly} is $O(\frac{n}{k^2-k})^{2k^2-2k})\cdot (O(n) + O(k^8)) = O(\frac{n^{2k^2-2k+1}}{(k^2-k)^{2k^2-2k}})$.
\end{proof}

It is worth noting that the analysis of Algorithm~\ref{algo:PDS_k-poly} can be applied to improve the algorithm proposed in~\cite{ElMallah93}, which determines the minimum dominating set of $k$-polygon graphs in $O(n^{4k+3})$ time. The correctness and time complexity analysis of Algorithm~\ref{algo:DS_k-poly} are similar to those of Algorithm~\ref{algo:PDS_k-poly} and are therefore omitted here.

\begin{algorithm}[htb]
	\caption{Determining the minimum dominating set on $k$-polygon graphs}\label{algo:DS_k-poly}
	\begin{algorithmic} [1]
		\baselineskip 14pt
		\REQUIRE a $k$-polygon graph $G$.
		\ENSURE  a minimum dominating set of $G$.
		
		\STATE Let $J$ be the intersection model of $G$;
  
		\FOR{each possible $\B$}
		  \FOR{each pair of sides $S_i$ and $S_j$ of $J$}
                \STATE Determine a minimum set $D'_{ij}\subseteq J_{ij}$ such that $J'_{ij}\subseteq N_J[D'_{ij}]$ while $J[D'_{ij}]$ contains a perfect matching;
            \ENDFOR

            \IF{$N_J[D'\cup \B] = J$}
                \STATE Let $D \leftarrow D'\cup \B$;
            \ENDIF
		\ENDFOR

		\RETURN the minimum set $D$ in Step~(7).
	\end{algorithmic}
\end{algorithm}

\begin{theorem} \label{theorem:DS_k-poly}
    Algorithm~\ref{algo:DS_k-poly} finds a minimum dominating set of $G$ in $O(\frac{n^{2k^2-4k+1}}{(k^2-k)^{k^2-k} (k^2-3k)^{k^2-3k}})$ time.
\end{theorem}
\begin{proof}
    The correctness of Algorithm~\ref{algo:PDS_k-poly} is ensured since at least one minimum dominating set $D$ with the outer boundary $\B$ has been explored. Furthermore, one can verify that $J_{ij}$ is exactly a permutation graph and, by Lemma~\ref{lemma:inner_boundary}, for each $j\in J'_{ij}$, $j\notin N_J[J-M_{ij}]$. Thus, each $J'_{ij}$ can be replaced by the set $D'_{ij}$, determining in Step~$(4)$, while still ensuring that $D$ remains a minimum paired-dominating set.

    The time complexity is as below.
    We first discuss the number of every possible outer boundary $\B$ in Step~$(2)$. Let $n_{ij}$ is the cardinality of $J_{ij}$. If $i=j-1$, the number of possibilities to exhaustively explore $\B_{ij}$ is on the order of $O((n_{ij})^2)$. Otherwise, the number of possibilities to exhaustively explore $\B_{ij}$ is on the order of $O((n_{ij})^4)$. Therefore, the number of possible boundaries is given by $\prod_{i = j-1} {n_{ij}}^{2}\cdot \prod_{i \neq j-1} {n_{ij}}^{4}$. This simplifies to $\prod_{1 \leq i < j \leq k} {n_{ij}}^2 \cdot \prod_{i \neq j-1} {n_{ij}}^2$.

    By the AM-GM inequality, $\sqrt[C(k,2)]{\prod_{1\le i<j\le k} n_{ij}}\le \frac{n}{C(k,2)}$ and $\sqrt[C(k,2)-k]{\prod_{i\neq j-1} n_{ij}}\le  \frac{n}{C(k,2)-k}$, where $C(k,2) = \frac{k^2-k}{2}$ represents the combination number for $1\le i<j\le k$, the possibilities of boundaries is $O(\frac{n^{2k^2-4k}}{(k^2-k)^{k^2-k} (k^2-3k)^{k^2-3k}})$.
    In Step~$(4)$, determining all $D'_{ij}$ can be done in totaling $O(n)$ time by Lemmas~\ref{lemma:adjacent_O(n)_permutation} and~\ref{lemma:general_O(n)_permutation}.
    Hence, the time complexity of Algorithm~\ref{algo:PDS_k-poly} is $O(\frac{n^{2k^2-4k+1}}{(k^2-k)^{k^2-k} (k^2-3k)^{k^2-3k}})$.
\end{proof}

\subsection{Improved Algorithms}
\label{sectionalgo_2_for_k-polygon}
In this subsection, we improve the time complexity of Algorithm~\ref{algo:DS_k-poly} for domination problem on $k$-polygon graph $G$ into $O(n^{3k-5})$. Let the drawing $G_J = (V_J, E_J)$ is constructed based on the intersection model $J$ of $G$, where each vertex $v_i \in V_J$ for $1 \leq i \leq k$ is positioned in the same order as the corresponding side $S_i$ of the $k$-polygon, and the edge $(v_i, v_j) \in E_J$ exists if and only if there is a chord of $J$ in $c(S_i) \cap c(S_j)$. We call an edge $e$ a \emph{cut} if the vertices shared by $e$ are not adjacent and no other edge in $E_S$ intersects $e$. Our algorithm is as below.


\begin{algorithm}[htb]
	\caption{Determining the minimum dominating set on $k$-polygon graphs}\label{algo:DS_k-poly_improved}
	\begin{algorithmic} [1]
		\baselineskip 14pt
		\REQUIRE a $k$-polygon graph $G$.
		\ENSURE  a minimum dominating set of $G$.
		
		\STATE Let $J$ be the intersection model of $G$;
		\STATE Let $D^* \leftarrow J$ initially;
  
		\FOR{each chord set $D\in J$ where $|D|\le 3k-6$} \label{step:3k-6}
		  \FOR{$t$ from $1$ to $\frac{k}{2}$} \label{step:t_from_1_to_k/2}
            \FOR{$i$ from $1$ to $k-t$}
            	\STATE $j = i+t$;
                \STATE Determine a minimum $T_{ij}\subseteq J_{ij}$ such that $c(S_i)\cap c(S_{j}) - N_J[D]\subseteq N_J[T_{ij}]$, while maximizing the number of chords in $J - N_J[D]$ dominated by $T_{ij}$; \label{step:permutation_improved}
                \STATE Update $D \leftarrow D\cup T$;
            \ENDFOR
            \ENDFOR

            \IF{$|D| < |D^*|$}
                \STATE Let $D^* \leftarrow D$;
            \ENDIF
            
		\ENDFOR

		\RETURN $D^*$.
	\end{algorithmic}
\end{algorithm}


\begin{theorem} \label{theorem:DS_k-poly_fast}
    Algorithm~\ref{algo:DS_k-poly_improved} determines a minimum dominating set of $G$ in $O(n^{3k-5})$ time.
\end{theorem}
\begin{proof}
    It suffices to prove the correctness by showing a minimum dominating set $D$ can be constructed by the process. Let $D^*$ be a minimum dominating set and $G_{D^*} = (V_{D^*},E_{D^*})$ be a drawing of ${D^*}$. We first extract at most $3k-6$ chords from $D^*$, assuming they are selected precisely in Step~(\ref{step:3k-6}), forming the union $D_1 \cup D_2 \cup D_3$. The set $D_1$ consists of chords $c \in D^* \cap c(S_i) \cap c(S_j)$ whose leftmost endpoints lie on their respective sides $S_i$ and $S_j$ for each cut $(v_i, v_j) \in E_{D^*}$. When $j = i + 1$, we select only the leftmost endpoints lying on $S_i$ and unite them into $D_2$. For $D_3$, it contains those chords $c \in D^* \cap c(S_i) \cap c(S_j)$ for each $(v_i, v_j) \in E_{D^*}$ that is not a cut. 
    The proof that $|D_1 \cup D_2 \cup D_3| \le 3k-6$ will be provided subsequently.
      
    Let $D = D_1 \cup D_2\cup D_3$ in initial. In Step~(\ref{step:permutation_improved}), since $J_{ij}$ remains a chord set of a $2$-polygon graph, which is equivalent to a permutation graph, one can apply the linear-time algorithm from~\cite{ElMallah93} to find a minimum dominating set $T_{ij}$ such that $J_{ij} \subseteq N_{J_{ij}}[D \cup T_{ij}]$ in $O(|J_{ij}|)$ time. After uniting $T_{ij}$ into $D$ for each previous iteration, since the leftmost endpoint of $S_i$ is decided by $D$, the dominating set obtained through this algorithm maximizes the number of chords in $J - N_J[D]$ that are dominated by $T_{ij}$. One can verify that $(D^* - D_1\cup D_2)_{ij}$ can be replaced by $T_{ij}$ for each $i,j$ while still remaining a minimum dominating set. Note that the overall process runs for each iteration of Step~(\ref{step:3k-6}) is $O(n)$ time. 
    
    At this point, the correctness of Algorithm~\ref{algo:DS_k-poly_improved} holds if $|D_1 \cup D_2 \cup D_3| \le 3k-6$.
    First, it is clearly that $|D_2|\le k$. Suppose that $D_1$ corresponds to $\alpha$ cuts in $G_{D^*}$. 
    It suffices to show that $|D_3| \le 2k-2\alpha-6$.  
    Let the set $D'$, included within the chord set determined in Step~(\ref{step:permutation_improved}), represents the chords $c \in D^* \cap c(S_i) \cap c(S_j)$ whose rightmost endpoints lie on their respective sides $S_i$ and $S_j$ if $(v_i, v_j)$ is a cut in $E_{D^*}$, or on $S_j$ if $j = i + 1$. One can see that the sole role of $D_3$ is to dominate the chords in $c(S_i)  \cap c(S_j) - N_J[D_1 \cup D']$ when $(v_i, v_j) \in E_{D^*}$ is not a cut. To achieve this, we actually have an alternative. Let $N_t$ be a maximum non-intersecting edge set of $G_{D_3} = (V_{D_3}, E_{D_3})$. For each $(v_i, v_j) \in E_{D_3}$, we can select an arbitrary chord $c \in c(S_i) \cap c(S_j)$ and pair it with another chord $c'$ that shares no common side with $c$. This alternative approach ensures that the chosen chords also dominate $c(S_i) \cap c(S_j) - N_J[D_1 \cup D']$ when $(v_i, v_j) \in E_{D^*}$ is not a cut. Moreover, since $G_{D^*}$ has at most $k-3$ cuts, the alternative approach selects at most $2(k-3-\alpha)$ chords. Therefore, $|D_3| \le 2k-6-2\alpha$. Combining this with $|D_1| \le 2\alpha$ and $|D_2| \le k$, we conclude that $|D_1 \cup D_2 \cup D_3| \le 3k-6$. This result confirms the correctness of Algorithm~\ref{algo:DS_k-poly_improved}. Moreover, the time complexity of Algorithm~\ref{algo:DS_k-poly_improved} is $O(n^{3k-6})*O(n) = O(n^{3k-5})$.
\end{proof}

Moerover, combining the result of Lemma~\ref{theorem:DS_k-poly_fast} and the technique of~\cite{Kratsch97}, we can obtain an $O(n^{3k-5})$ algorithm for total domination problem. 

\begin{theorem} \label{theorem:TDS_k-poly_fast}
	A minimum total dominating set of $G$ can be determined in $O(n^{3k-5})$ time.
\end{theorem}

\smallskip
\section{Conclusion and Future Work} \label{section:conclusion}
In this paper, we present two main results concerning paired domination on circle graphs and $k$-polygon graphs. For the result on circle graphs, we successfully demonstrated the NP-completeness of the paired-domination problem. Importantly, our unique reduction from the Hamiltonian path problem sets our approach apart from existing research methodologies. We encourage readers to employ this novel strategy in investigating the computational complexity of domination and its variants on circle graphs, as well as other classes of graphs. Notably, prior work by Damian and Pemmaraju~\cite{Damian06} has established the APX-hardness of the domination and total domination problems on circle graphs. Therefore, we conjecture that the paired-domination problem also maintains APX-hardness on circle graphs, presenting a promising avenue for exploration using our innovative reduction approach. Additionally, leveraging the $(2+\varepsilon)$-approximation algorithm for the domination problem from~\cite{Damian02}, a $(4+\varepsilon)$-approximation algorithm for the paired-domination problem on circle graphs can be readily derived for $\varepsilon>0$. Future research can focus on devising approximation algorithms with factors below $4$, representing an exciting direction for further investigation.

Next, we present another main result on $k$-polygon graphs. We first demonstrate the algorithm for paired-domination and domination problem, with time complexity $O(n(\frac{n}{k^2-k})^{2k^2-2k})$ and $O(\frac{n^{2k^2-4k+1}}{(k^2-k)^{k^2-k} (k^2-3k)^{k^2-3k}})$, respectively. Moreovre, we further improve the previous algorithm for domination problem into $O(n^{3k-5})$. We anticipate that the time complexity for the paired-domination one can also be significantly reduced in future research. Additionally, as circle graphs are a superclass of $k$-polygon graphs, the algorithm for solving domination problems and its variants on $k$-polygon graphs might provide a feasible direction for designing approximation algorithms on circle graphs.

\bibliographystyle{abbrv}

\end{document}